\documentclass{article}
\usepackage{amsmath,amsfonts, amsthm,amssymb,fullpage}
\newtheorem{conjecture}{Conjecture}
\newtheorem{proposition}{Proposition}
\newtheorem{claim}{Claim}
\newcommand{\F}{\mathbb{F}}
\newcommand{\inset}[1]{\left\{#1\right\}}
\usepackage[pdftex]{hyperref}

\title{A Note on the Permuted Puzzles Toy Conjecture}
\author{Keller Blackwell\thanks{Stanford University.  \texttt{kellerb@stanford.edu}} \ and Mary Wootters\thanks{Stanford University.  \texttt{marykw@stanford.edu}}}

\begin{document}
\maketitle

\begin{abstract}
In this note, we show that a ``Toy Conjecture'' made by (Boyle, Ishai, Pass, Wootters, 2017) is false, and propose a new one.  Our attack does not falsify the full (``non-toy'') conjecture in that work, and it is our hope that this note will help further the analysis of that conjecture.  Independently, (Boyle, Holmgren, Ma, Weiss, 2021) have obtained similar results.
\end{abstract}

\section{Introduction}
Recently, two independent works~\cite{BIPW17,CHS17} proposed a notion
called \em oblivious locally decodable codes, \em (OLDCs), motivated by applications in private information retrieval (PIR).\footnote{The OLDC terminology is from \cite{BIPW17}; in \cite{CHS17}, the corresponding notion is designated-client doubly-efficient PIR.}  These works gave candidate constructions of OLDCs based on a new conjecture regarding the hardness of distinguishing a uniformly random set of points from a permutation of local-decoding queries in a Reed-Muller code.  In order to encourage study of this conjecture, \cite{BIPW17} proposed a simplified ``Toy Conjecture,'' which we reproduce below as Conjecture~\ref{conj}.  In this note, we show that this Toy Conjecture is false, and propose a new one that is resistant to our attack.  We note that this does not refute the full conjecture of \cite{BIPW17,CHS17}.

Independently of this note, Boyle, Holmgren, Ma and Weiss have also established that the Toy Conjecture is false~\cite{BHMW21}, and have proposed a new Toy Conjecture.

\section{The Toy Conjecture and an Attack}
The Toy Conjecture of \cite{BIPW17} is the following.
\begin{conjecture}[Toy Conjecture 4.6 in \cite{BIPW17}]\label{conj}
Let $\F$ be a finite field with $|\F| = q \approx \lambda^2$.  Let $p_1, \ldots, p_m$ be uniformly random polynomials of degree at most $\lambda$ in $\F[X]$, for $m = \lambda^{100}$.  Let $q_1, \ldots, q_m$ be uniformly random functions from $\F$ to $\F$.  
Let $\pi \in S_{\F \times \F}$ be a uniformly random permutation.  Then the following two distributions are computationally indistinguishable:
\begin{itemize}
\item[(1)] $(S_1, \ldots, S_m)$, where $S_i = \inset{ \pi(x, p_i(x)) \ :\ x \in \F }$.
\item[(2)] $(T_1, \ldots, T_m)$, where $T_i = \inset{ \pi(x, q_i(x)) \ :\ x \in \F }$.
\end{itemize}
\end{conjecture}

To show that this conjecture is false, we give below an efficient algorithm to distinguish the two distributions.

\begin{itemize}
\item Input: $(U_1, \ldots, U_m)$, where $U_i \subset \F \times \F$
\item Construct that matrix $M \in \F^{m \times (\F \times \F)}$ that is given by 
$$M_{i, (\alpha, \beta)} = \mathbf{1}[ (\alpha, \beta) \in U_i ].$$
\item If $\mathrm{rank}(M) < q^2 - q + 1$, output ``Case (1).''
\item Otherwise, output ``Case (2).''
\end{itemize}

\begin{proposition} The algorithm above correctly distinguishes between Case (1) and Case (2) with high probability.  More precisely, in Case (1), the algorithm returns ``Case (1)'' with probability $1$.  In Case (2), the algorithm returns ``Case (2)'' with probability at least $1 - e^{-\lambda^{97}}$ over the choice of the functions $q_i$.
\end{proposition}
\begin{proof}
To prove that this algorithm is correct, we will show that under distribution (2), the matrix $M$ has rank exactly $q^2 - q + 1$ with high probability; while under distribution (1), the matrix has rank strictly less than that.

Suppose that $f_1, \ldots, f_m$ are the functions that drawn (either $f_i = p_i$ in case (1), or $f_i = q_i$ in case (2)).
Let $A \in \F^{m \times (\F \times \F)}$ be the matrix that is given by
$$ A_{i, (\alpha, \beta)} = \mathbf{1}[f_i(\alpha) = \beta].$$
Notice that $M$ is a column permutation of $A$, so $\mathrm{rank}(A) = \mathrm{rank}(M)$.  
Thus, to show that the algorithm above is correct, it suffices to study the random of $A$ in cases (1) and (2).
In the following, we let $A_i$ denote the $i$'th row of $A$.

Let 
$$K = \inset{ v\in \F^{\F \times \F}\ :\ v_{(\alpha,\beta)} = w_\alpha \text{ for some $w_\alpha \in \F^\F$ so that } \sum_{\alpha \in \F} w_\alpha = 0}.$$
Notice that $K$ is a subspace of $\F^{\F \times \F}$ and that $\mathrm{dim}(K) = q-1$.
\begin{claim}\label{cl:1}
Suppose that case (2) holds, so $f_i = q_i$ is a uniformly random function.  Then $\mathrm{Ker}(A) = K$ with probability at least $1 - \lambda^{-97}$ over the choice of the functions $f_i$.
\end{claim}
\begin{proof}
First, observe that $K \subseteq \mathrm{Ker}(A)$, since for any $i \in [m]$,
\[ \sum_{\alpha, \beta \in \F} A_{i, (\alpha,\beta)} v_{(\alpha, \beta)} = \sum_{\alpha, \beta \in \F} \mathbf{1}[f_i(\alpha) = \beta] w_\alpha = \sum_{\alpha \in \F} w_\alpha = 0. \]
On the other hand, suppose that $v \not\in K$.  If $v_{(\alpha, \beta)} = w_\alpha$ for some $w \in \F^\F$ so that $\sum_{\alpha \in \F} w_\alpha \neq 0$, then clearly $v \not\in \mathrm{Ker}(A)$.  So suppose that $v_{(a,b)} \neq v_{(a,b')}$ for some $a,b,b' \in \F$.
Then let
\[ X_i = \sum_{\alpha \neq a} \sum_{\beta \in \F} A_{i, (\alpha, \beta)} v_{(\alpha,\beta)}.\]
This is a random variable over the choice of $f_i$.  Now, for any $i$, and for any $x \in \F$,
\[ \Pr\left[ A_i^T v = 0 \mid X_i = x \right] = \Pr\left[ \sum_{\beta \in \F} A_{i,(\alpha,\beta)} v_{(\alpha,\beta)} = -x \right] = \Pr\left[ v_{a, f_i(a)} = -x \right], \]
where again the probability is over the choice of $f_i$. 
Since $v_{(a,b)} \neq v_{(a,b')}$, there is at least a $1 - 1/q$ chance that $v_{a, f_i(a)} \neq -x$, if $f_i = q_i$ is a uniformly random function.  Thus, for all $i \in [m]$ and for all $x \in \F_q$,
\[ \Pr\left[ v_{(a, f_i(a))} \neq -x \right] \leq 1 - 1/q. \]
This implies that for all $i \in [m]$, 
\[ \Pr\left[ A_i^T v = 0 \right] = \sum_{x \in \F} \Pr[X_i = x] \Pr[ A_i^T v = 0 \mid X_i = x ] \leq 1 - 1/q. \]
By the independence of the $f_i$,
\[ \Pr[ A_i^T v = 0 \forall i \in [m] ] \leq \left( 1 - 1/q \right)^m \leq e^{-m/q}. \]
By the union bound over all $v$ of this form,
\[ \Pr[ \exists v \not\in K, A_i^T v = 0 \forall i \in [m] ] \leq q^{q^2} e^{-m/q} \leq e^{q^2 \log q - m/q } \leq e^{-\lambda^{97}}, \]
using the choice of $q \approx \lambda^2$ and $m = \lambda^{100}$.
\end{proof}

This establishes that, in case (2), with probability at least $e^{-\lambda^{97}}$, $A$ has rank
\[ \mathrm{rank}(K) = q^2 - \mathrm{dim}(K) = q^2 - q + 1. \]

On the other hand, in case (1), $A$ has kernel vectors that are not in $K$.  One example is the vector $v \in \F^{\F^2}$ given by $v_{(\alpha,\beta)} = \beta$.  To see that $v \in \mathrm{Ker}(A)$, when $f_i = p_i$ is a polynomial of degree $\lambda < q-1$, observe that
\[ A_i^T v = \sum_{\alpha, \beta \in \F} \mathbf{1}[p_i(\alpha) = \beta] \cdot \beta = \sum_{\alpha \in \F} p_i(\alpha) = 0, \]
where in the final equality we have used the fact that $\sum_{\alpha \in \F} \alpha^c = 0$ for any $0 \leq c < q-1$.  

This establishes that, in case (1), $\mathrm{Ker}(A) \supsetneq K$, which implies that $A$ has rank
\[ \mathrm{rank}(K) = q^2 - \mathrm{dim}(K) < q^2 - q + 1. \]

This shows that the algorithm above correctly distinguishes between cases (1) and (2), with probability at least $1 - e^{-\lambda^{97}}$.

\end{proof}

\section{A New Toy Conjecture}
We note that the attack above does not work if the evaluation points for the $f_i$ are a random subset resampled each time (which is indeed the case for the more general permuted puzzles conjecture of \cite{BIPW17}).  Thus, we propose the following replacement toy conjecture:
\begin{conjecture}[New Toy Conjecture]
Let $\F$ be a finite field with $|\F| = q \approx \lambda^2$.  Let $p_1, \ldots, p_m$ be uniformly random polynomials of degree at most $\lambda$ in $\F[X]$, for $m = \lambda^{100}$. Let $q_1, \ldots, q_m$ be uniformly random functions from $\F$ to $\F$.  
Let $\pi \in S_{\F \times \F}$ be a uniformly random permutation.  
Let $\Omega^{(1)}, \ldots, \Omega^{(m)} \subset \F$ be independent uniformly random sets of size $100 \cdot \lambda$.
Then the following two distributions are computationally indistinguishable:
\begin{itemize}
\item[(1)] $(S_1, \ldots, S_m)$, where $S_i = \inset{ \pi(x, p_i(x)) \ :\ x \in \Omega^{(i)} }$.
\item[(2)] $(T_1, \ldots, T_m)$, where $T_i = \inset{ \pi(x, q_i(x)) \ :\ x \in \Omega^{(i)} }$.
\end{itemize}
\end{conjecture}

\section*{Acknowledgements} We thank Elette Boyle, Justin Holmgren, Fermi Ma and Mor Weiss for helpful conversations and for pointing out typos in an earlier version of this note.

\bibliographystyle{alpha}
\bibliography{refs}
\end{document}